\documentclass[reqno]{amsart}

\usepackage{amsmath}
\usepackage{amssymb}
\usepackage{amsthm}
\usepackage{bm}
\usepackage{eucal}
\usepackage{fullpage}
\usepackage{graphicx}
\usepackage{lineno}
\usepackage{mathabx}
\usepackage{mathrsfs}
\usepackage[numbers,sort&compress,square]{natbib}
\usepackage{setspace}
\usepackage{subfig}
\usepackage{xcolor}
\usepackage[all,cmtip]{xy}

\makeatletter
\@fpsep\textheight
\makeatother

\title{Effects of motion in structured populations}
\author{Madison S. Krieger$^{\ast}$}
\author{Alex McAvoy$^{\ast}$}
\author{Martin A. Nowak}
\thanks{$^{\ast}$M. S. K. and A. M. contributed equally to this study. Corresponding author: Martin A. Nowak (\texttt{martin\_nowak@harvard.edu})}

\theoremstyle{definition}
\newtheorem{lemma}{Lemma}

\begin{document}

\allowdisplaybreaks

\begin{abstract}
In evolutionary processes, population structure has a substantial effect on natural selection. Here, we analyze how motion of individuals affects constant selection in structured populations. Motion is relevant because it leads to changes in the distribution of types as mutations march toward fixation or extinction. We describe motion as the swapping of individuals on graphs, and more generally as the shuffling of individuals between reproductive updates. Beginning with a one-dimensional graph, the cycle, we prove that motion suppresses natural selection for death-birth updating or for any process that combines birth-death and death-birth updating. If the rule is purely birth-death updating, no change in fixation probability appears in the presence of motion. We further investigate how motion affects evolution on the square lattice and weighted graphs. In the case of weighted graphs we find that motion can be either an amplifier or a suppressor of natural selection. In some cases, whether it is one or the other can be a function of the relative reproductive rate, indicating that motion is a subtle and complex attribute of evolving populations. As a first step towards understanding less restricted types of motion in evolutionary graph theory, we consider a similar rule on dynamic graphs induced by a spatial flow and find qualitatively similar results indicating that continuous motion also suppresses natural selection.
\end{abstract}

\maketitle

\section{Introduction}
Population structure is of great interest in ecology and evolution. Traditionally, subdivisions have been used to model the geographical patches in which the individuals of a population reside \citep{pollak:JAP:1966,maruyama:GR:1970,christiansen:AN:1974,whitlock:G:1997,wakeley:TPB:2004}. These patches can account for environmental heterogeneity in reproductive rates, clustering and spatial segregation, speciation, and the effects of migration \citep{nagylaki:JMB:1980,mcpeek:AN:1992,gavrilets:PE:2002,church:E:2002,hill:AN:2002,whitlock:G:2005,fu:JSP:2012}. The incorporation of even simple population structures has substantially improved the descriptive power of mathematical models of evolution. As a result, population structure has become a cornerstone of evolutionary analysis.

Evolutionary graph theory is a powerful framework to describe population structure. On an evolutionary graph, individuals occupy the vertices, and edges specify interactions between individuals or dispersal patterns \citep{lieberman:Nature:2005,ohtsuki:Nature:2006,ohtsuki:JTB:2007,broom:PRSA:2008,holland:Nature:2008,houchmandzadeh:NJP:2011,broom:JSTP:2011,hadjichrysanthou:DGA:2011,broom:TF:2013,debarre:NC:2014,allen:Nature:2017}. Graphs can represent spatial structures, social networks, or organizational hierarchies of biological populations. By allowing edge weights to vary and the use of multiple graphs to model different parts of a population structure (such as separate interaction and dispersal neighborhoods), evolutionary graph theory allows one to model a variety of realistic population structures. Here, edge weights may represent spatial heterogeneity or ecological interactions that alter the probabilities that the offspring of one individual will come to occupy a particular adjacent space. For example, the unstructured (or well-mixed) population is given by a complete graph with identical weights; many spatial models are represented by graphs that are regular grids; and island models can be described by weighted graphs, where each island is a complete graph and the connections between islands have smaller weights.

Most evolutionary models involve some form of motion. In a spatially-structured population, reproduction itself is a type of motion because offspring can disperse to locations that are different from that of their parents. Migration, which need not be tied to reproductive events, has been studied in both population genetics \citep{wright:G:1946,kimura:G:1964,heino:AN:2001,vuilleumier:TPB:2010} and evolutionary game theory \citep{miekisz:EGT:2008,ladret:JTB:2008,tarnita:PNAS:2009,hauert:JTB:2012,cressman:PNAS:2014,pichugin:JTB:2015} and represents another type of motion. Other evolutionary models, such as those involving shift updating \citep{allen:JTB:2012b}, repulsion \citep{pavlogiannis:SR:2015}, spatial games \citep{nowak:Nature:1992,nowak:IJBC:1994,nowak:PNAS:1994,lindgren:PD:1994,ranta:EP:2005}, stepping-stone models \citep{korolev:RMP:2010}, expansion and growth \citep{hallatschek:E:2010,pigolotti:TPB:2013}, and spatial self-structuring \citep{boerlijst:PD:1991}, also contain some form of motion. In fact, the motion implicit in many evolutionary models accounts for one of the reasons for the profound influence of population structure on evolutionary dynamics.

Our focus here is on an abstract type of motion, inspired by these works, that has not yet been considered within evolutionary graph theory. To begin, we study motion idealized as a swapping (or shuffling) of individuals that occurs independent of reproductive events. Our motion changes the distribution of traits within a given population structure. This initial approach is distinct from dynamic evolutionary graphs, in which edges can be created or destroyed, for example when individuals of certain types preferentially create connections with one another \citep{pacheco:PRL:2006,pacheco:JTB:2008,wardil:SR:2014,wardil:PLOSONE:2016}, when the graph structure is determined by set membership \citep{tarnita:PNAS:2009}, or when game players imitate the social networks of successful players \citep{cavaliere:JTB:2012,cavaliere:SR:2016}. It is also distinct from models where individuals can move into unoccupied spaces \citep{basanta:TEPJB:2008,basanta:MB:2011,lee:CR:2011,waclaw:Nature:2015,manem:PLOSONE:2015}. Once we establish basic results for the effect of swaps on fixation probabilities, we extend our attention to dynamic graphs, in which the possible individuals which can interact via reproductive dispersal are changing at all times. 

We use the Moran process \citep{moran:MPCPS:1958}, which acts on a finite population of fixed size $N$. In its most basic realization, every individual is one of two types (mutant or resident). At each time step, one individual is chosen to reproduce and one individual dies. These events could occur in either order. For birth-death (BD) updating, an individual is selected from the population for reproduction with a probability proportional to its fitness, and the offspring replaces a random neighbor. For death-birth (DB) updating, the order is reversed, but with fitness once again affecting the birth event only. These update rules generate simple evolutionary processes and are convenient models for studying the effects of motion on natural selection.

We begin by studying the effect of motion in one dimension. The adoption of periodic boundaries leads to the cycle \citep{ohtsuki:PRSB:2006,allen:JTB:2012b}. Let $\rho_{1}\left(r\right)$ be the probability that an arbitrarily placed mutant with relative fitness $r$ takes over a resident population on the cycle. This ``fixation probability" provides a convenient way to quantify the evolutionary success of a mutant type \citep{robertson:PRSB:1960}. Let $\rho_{1}^{\ast}\left(r\right)$ be the fixation probability after introducing motion on the cycle. We prove that any motion on the cycle, even motion that changes arbitrarily or is repeated many times after each update, is a suppressor of natural selection, which means that $\rho_{1}^{\ast}\left(r\right)\leqslant\rho_{1}\left(r\right)$ when $r>1$ and $\rho_{1}^{\ast}\left(r\right)\geqslant\rho_{1}\left(r\right)$ when $r<1$. This result holds for DB updating and mixtures of BD and DB in which birth and death events are not strictly ordered \citep{zukewich:PLoSONE:2013}. The only case where suppression due to motion does not occur is if the update rule is purely BD. That motion suppresses natural selection is therefore a universal property of evolution on one-dimensional structures.

After studying motion on the cycle, we consider more complicated graphs. On the square lattice, random swapping of individuals still leads to results consistent with suppression of selection. This behavior also holds for a particular example of a dynamic graph. However, motion on weighted, directed, or heterogeneous graphs is found to have more nuanced and complicated effects. For example, motion can amplify selection, suppress selection, or even act in such a way as to either amplify or suppress selection depending on the value of the \textit{selective advantage} $r$. These findings are similar to the conclusions drawn from studying the effects of population structure on selection, where spatial structure can amplify or suppress selection \citep{nowak:PNAS:2003,adlam:PRSA:2015,jamiesonlane:JTB:2015,galanis:amplifiers:2015,giakkoupis:amplifiers:2016,goldberg:amplifiers:2016,pavlogiannis:SR:2017}, or even fall into another category altogether \citep{hindersin:PLOSCB:2015}. Motion, on the other hand, is even more intricate because it can arise in many forms and its effects on selection depend heavily on the underlying population structure. Here, we highlight these issues and demonstrate how the effects of motion on selection are contingent on population structure.

\section{Modeling motion on graphs}
Our model of motion is one where individuals exist in vertices on a graph, which is held fixed. The individuals themselves are swapped from one vertex to another, and this represents motion. We call a series of one or more swaps a shuffle. A shuffle of the population is a permutation, $\pi\in\mathfrak{S}_{N}$, where $\mathfrak{S}_{N}$ is the symmetric group on $N$ letters, that is the set of all permutations of $\left\{1,\dots ,N\right\}$. Given a state of the system $\left(s_{1},\dots ,s_{N}\right)$, where $s_{i}$ is either $A$ (mutant) or $B$ (resident) for each $i=1,\dots ,N$, the post-shuffle configuration is $\left(s_{\pi\left(1\right)},\dots ,s_{\pi\left(N\right)}\right)$. The type of individual $i$ is $s_{i}$ before the shuffle and $s_{\pi\left(i\right)}$ after the shuffle, effectively redistributing the types within the fixed population structure.

A shuffle can disturb the boundaries between groups of one type and groups of another. For example, motion on the cycle can break a single cluster into groups of smaller clusters or even into isolated individuals (see Fig. \ref{fig:motionOnCycle}).

\begin{figure}
\centering
\includegraphics[scale=0.55]{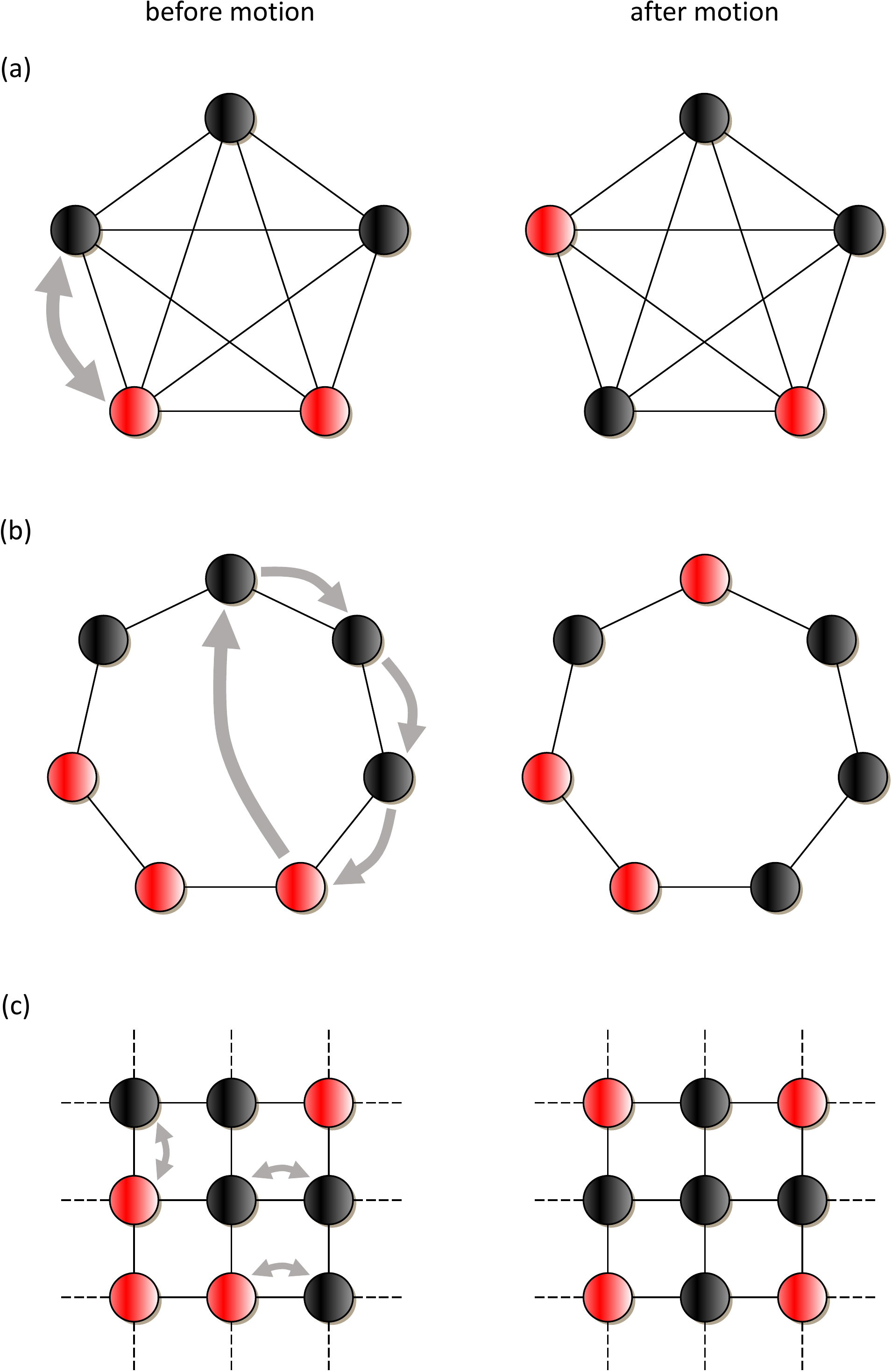}
\caption{\small Motion in structured populations. We study motion as a process in which individuals change their location, but the overall structure of the population is preserved. Individuals of two types, resident (black) and mutant (red) are arranged on a graph. (a) On the complete graph (sometimes called unstructured or well-mixed population) any two individuals are neighbors. Motion does not change the configuration of the population and therefore has no effect on evolutionary dynamics. (b) The cycle is the simplest one-dimensional population structure. Motion can break up clusters of individuals and therefore modify the evolutionary process. (c) The square lattice is a regular two-dimensional structure. Again, motion can change evolutionary dynamics. In the \textbf{Supporting Information} we also discuss heterogeneous motion, which can affect different parts of the graph in different ways.\label{fig:motionOnCycle}}
\end{figure}

While a shuffle captures the concept of motion, such motion can be stochastic rather than deterministic. In order to capture a more generic concept of motion, we consider a probability distribution, $\mu$, over the set of possible shuffles. The probability that the population is shuffled according to $\pi\in\mathfrak{S}_{N}$ at a given time step is $\mu\left(\pi\right)$, which we take to be time-independent. One reason that this formulation is convenient is that the notion of shuffling is closed under composition. That is, if $\mu$ and $\nu$ are such distributions, then the sequence of $\mu$ followed by $\nu$, $\nu\circ\mu$, is again a well-defined distribution over shuffles (see \textbf{Supporting Information} for an expression for $\nu\circ\mu$). This property extends to sequences of shuffles of any length. For example, if we consider only nearest-neighbor swaps, the composition property allows us to easily find the probability of getting from one arrangement to another within a prescribed number of such swaps. As a result, we can model arbitrary motion of types within a graph via a single stochastic shuffle in each time step.

Even on a homogeneous structure, motion can destroy this homogeneity by affecting different parts of the structure in different ways. In Fig. \ref{fig:heterogeneousMotion}, we give an example of motion on half of the cycle that results in a dependence of fixation probability on initial location. In such a scenario, if $\rho_{1,i}$ denotes the fixation probability of a single mutant initially at location $i$, then one can quantify the effects of motion using the average fixation probability,
\begin{linenomath}
\begin{align}\label{eq:averageFP}
\rho_{1} &= \frac{1}{N} \sum_{i=1}^{N} \rho_{1,i} .
\end{align}
\end{linenomath}
Henceforth, by ``fixation probability" we mean the average as defined by Eq. (\ref{eq:averageFP}).

\section{Motion on a one-dimensional structure}
The simplest structure on which to examine the role of motion is a population of $N$ individuals distributed on a cycle, with each individual having two unique neighbors (see Fig. \ref{fig:motionOnCycle}). This structure is equivalent to a line, or one-dimensional lattice, where the endpoints are joined to one another to avoid boundary effects \citep{allen:JMB:2014}. The cycle is also particularly simple due to its spatial homogeneity arising from rotational symmetry.

We study the process where a single mutant of fitness $r$ (relative to the resident type, which has fitness $1$) arises in a population. The state of the population changes through a sequence of discrete update steps, with each update step consisting of the following sequence: First, an individual is chosen with uniform probability to move and is swapped with any individual within $d$ positions on the cycle (chosen with uniform probability). This swapping is repeated $J$ times. After the swapping, reproduction and replacement take place with either birth-death (BD) or death-birth (DB) updating. Fig. \ref{fig:motionOnCycle} depicts this process and demonstrates how motion can disrupt clusters of mutants. Note that while $d$ can take on any value in our model, higher values of $d$ may not be biologically relevant, as they represent much farther dispersals with no effect on the intervening individuals between the two individuals being swapped. In a natural setting, motions which interchange two individuals at a large distance without disturbing the rest of the population structure may be forbidden; for example, in a fluid flow, continuity of flow guarantees that such interchanges are impossible. Therefore, when we perform numerical simulations, we consider small values of $d$, which could be taken to represent swaps which naturally occur due to thermal or driven noise, causing minor disruptions in population structure. However, large values of $d$ are interesting from a theoretical perspective, and we discuss them further in the \textbf{Supporting Information}. In particular, there are threshold values of $d,J$ beyond which no further effect is seen on the fixation probability, representing the maximum amount of mixing in the population due to this type of motion.  

When there is no motion, the system's state space is effectively one-dimensional. The state of the system is uniquely specified by the number of mutants, $m$, since they are all adjacent to one another in a single cluster. The position of this cluster is irrelevant, due to the rotational symmetry of the cycle. Any cluster of size $m$ looks the same as any other cluster of size $m$ after a rotation of the lattice. However, since motion can break up this cluster, we must consider the impact that multiple mutant clusters have on evolutionary dynamics. If it happens that different arrangements of $m$ mutants have different chances of reproducing and replacing the resident, then we must know which of the $2^{N}$ states (modulo the rotational symmetry, which drastically decreases the number of distinct states) the system assumes to predict its subsequent behavior. 

For any possible arrangement of the cycle, we consider the transition probabilities of the evolutionary process. These transitions give the probabilities that one type reproduces and another dies, while the population size remains constant. The transition probabilities are assigned to any single state. If mutants are favored in states with more (versus fewer) clusters, then motion is expected to enhance the mutant's fixation probability due to its tendency to break up clusters. If the fracturing of clusters has the opposite effect on transition probabilities, then the fixation of the mutant is suppressed.

Once the type of motion is prescribed, the stochastic process is completely determined. Given the current state, motion assigns a probability to each possible configuration of $m$ mutants, and then the reproductive process assigns a probability to each possible configuration of $m-1$, $m$, or $m+1$ mutants. However, in practice, it is unwieldy to calculate the fixation probability by summing over transition probabilities given by such a large number of states. Our approach is then to first understand the transition probabilities in the reproductive step, that is, what effect rearranging mutants has on their short-term chances to reproduce whilst the resident type dies. This approach turns out to be sufficient to conclude that motion suppresses natural selection on the cycle, as we show with a formal proof. 

\subsection{Birth-death (BD) updating}
We begin with a simple analysis of the transition probabilities of all possible configurations of $m$ mutants among $N$ individuals. The state of the cycle can be interpreted as a binary string of $A$ and $B$, where $A$ indicates mutant, while $B$ indicates resident. The strings have periodicity $N$, which is the total number of individuals. The full state space of the system is of size $2^{N}$, though not all of these states are distinct due to the rotational symmetry of the cycle. For BD updating, an individual is selected for reproduction with probability proportional to fitness. If there are $m$ mutants, the chance that any particular mutant is selected for reproduction is $r/\left(mr + N-m\right)$, and the chance that any resident type is selected for reproduction is $1/\left(m r + N-m\right)$, where $r$ is the selective advantage. The individual selected to reproduce replaces either the individual on the left or the right with its offspring, each with probability $1/2$. 

Suppose that $m$ mutants are distributed among $c$ clusters (of any size). Let $p^{+}$ (resp. $p^{-}$) be the probability that mutants increase (resp. decrease) in abundance in the next time step, i.e.
\begin{linenomath}
\begin{subequations}
\begin{align}
p^{+} &= \frac{rc}{m r + N-m} ; \\
p^{-} &= \frac{c}{m r + N-m} .
\end{align}
\end{subequations}
\end{linenomath}
A useful quantity for understanding the short-term dynamics is the \textit{forward bias}, $\gamma :=p^{+}/p^{-}$. In this case,
\begin{equation}
\gamma = r . \label{eq:forwardBiasBD}
\end{equation}
Here, the forward bias is the same as that of the Moran process described in the Introduction \citep{lieberman:Nature:2005}. Since motion on a graph can change $c$ but not $\gamma$, it has no effect on fixation probability under BD updating (see \textbf{Supporting Information}).

However, it is easily seen that $1-p^{+}-p^{-}$ is linearly dependent on $c$, the number of clusters. When there are more clusters, the process moves faster toward the absorbing states $m=0$ and $m=N$. The process is slowest when $c=1$, which means that the mutants are arranged in a single cluster on the cycle. Consequently, through its effects on $c$, motion accelerates evolution, whether the new type fixes or dies out. In fact, when an algorithm is applied that maximizes $c$ at every update step, the novel type reaches fixation faster even than in the unstructured, well-mixed population (see Fig. \ref{fig:ALGORITHM} and \textbf{Supporting Information}). 

\subsection{Death-birth (DB) updating}
We now consider the DB update rule. In this process, one individual is selected to die with uniform probability, and a neighbor replaces it with probability proportional to fitness. An isolated resident or mutant is therefore replaced immediately by the other type if selected for death. A vacancy created between a mutant and a resident is filled by the mutant with probability $r/\left(1+r\right)$.

For DB updating, clusters of different sizes have different transition probabilities. In particular, an isolated mutant is in more danger of dying out in the short term than a cluster of two or more mutants, where each has a chance to replace the other if one should die. If $x_{A}$ and $x_{B}$ are respectively the number of isolated $A$ (mutants) and $B$ (resident type) individuals in the current state, and $y_{A}$ and $y_{B}$ are respectively the number of clusters (of any size greater than one) of $A$ and $B$, then
\begin{linenomath}
\begin{subequations}
\begin{align}
q^{-} &= \frac{x_{A} + \frac{2}{1+r} y_{A}}{N} ; \\
q^{+} &= \frac{x_{B} + \frac{2r}{1+r} y_{B}}{N} ; \\
\gamma &= \frac{x_{B}+\frac{2r}{1+r} y_{B}}{x_{A}+\frac{2}{1+r} y_{A}} . \label{eq:gammaCDB}
\end{align}
\end{subequations}
\end{linenomath}
Here, we have used $q$ for transition probabilities instead of $p$ to distinguish DB updating from BD updating, so $\gamma := q^{+}/q^{-}$. This ratio of transition probabilities is plotted in Fig. \ref{fig:gamCycPlot}. The short-term properties of the process under DB updating are extremely state-dependent. We must consider the bias of a particular state of the system $s$, which we write as $\gamma(s)$. We show in \textbf{Supporting Information} that, for an advantageous mutant (those with $r>1$), $1\leqslant\gamma\left(s\right)\leqslant\gamma\left(s'\right)$ whenever $s'$ is a state whose mutants are the same in number as $s$ but are arranged in a cluster. Similarly, $\gamma\left(s'\right)\leqslant\gamma\left(s\right)\leqslant 1$ when the mutant is disadvantageous ($r<1$). Since motion on the cycle can disrupt clusters of mutants, it therefore suppresses natural selection (see \textbf{Supporting Information}). We plot the conditional fixation time of the mutant and the fixation probability in Fig. \ref{fig:DBprobplot}.

\begin{figure}
\centering
\includegraphics[width=0.8\columnwidth]{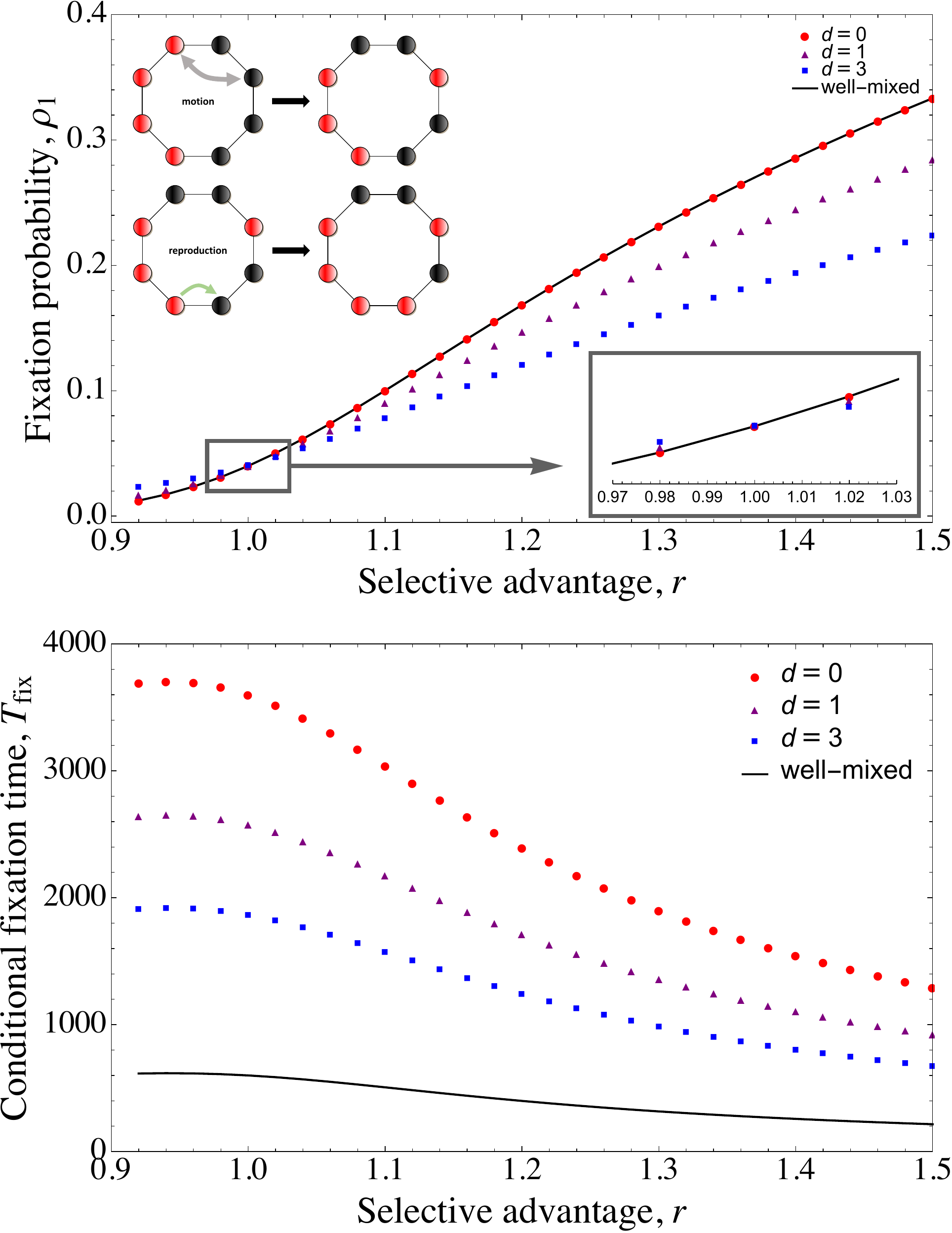}
\caption{\small Death-birth updating on a cycle with motion. Fixation probability (top) and conditional fixation time (bottom) are shown for a single mutant of relative fitness $r$ in a population of $N=25$, averaged over $10^{6}$ realizations. After each reproductive event, one individual is selected at random and swaps position with another random individual who is who is no more than $d$ steps away. The curves shown have $d=0$, which is equivalent to no motion (red circles), $d=1$ (purple triangles), and $d=3$ (blue squares). The well-mixed result for a population of the same size is given by the solid black line. We find that motion acts as a suppresser of natural selection (top) and accelerates the evolutionary process on the cycle (bottom).\label{fig:DBprobplot}}
\end{figure}

\subsection{Mixed BD and DB updating}
In nature, birth and death need not always be strictly ordered. Birth of an individual can come before death of another and, likewise, death of one can come before birth of another. We therefore consider a case in which, at each time step, a DB rule is used with probability $\delta$ and a BD rule is used with probability $1-\delta$ \citep{zukewich:PLoSONE:2013}. Classical BD and DB updating may be thought of as boundary cases occurring at $\delta =0$ and $\delta =1$, respectively.

For mixed BD and DB updating, the forward bias in a given state is
\begin{linenomath}
\begin{subequations}
\begin{align}
\gamma&= \frac{\left(1-\delta\right) p^{+} +\delta q^{+}}{\left(1-\delta\right) p^{-} +\delta q^{-}} .
\end{align}
\end{subequations}
\end{linenomath}
This expression is not a linear combination of the forward biases for BD and DB updating, Eqs. (\ref{eq:forwardBiasBD}) and (\ref{eq:gammaCDB}). Nevertheless, if a given arrangement $s'$ is a state with exactly $\left| s'\right| =\left| s\right|$ mutants arranged in a single cluster, then $\min\left\{\gamma\left(s',\delta\right),1\right\} \leqslant \gamma\left(s,\delta\right)\leqslant\max\left\{\gamma\left(s',\delta\right),1\right\}$. We establish these inequalities in \textbf{Supporting Information} and use them to show that motion suppresses selection for any $\delta\in\left[0,1\right]$.

\section{Motion on more complex graphs}
We first consider square lattices and then some examples of weighted graphs.

\subsection{Square lattices}
To confirm whether the suppression of selection seen with motion on the cycle extends to other structures, we simulated death-birth (DB) updating on the two-dimensional square lattice. The reproduction neighborhood is von Neumann, such that the sites to the north, east, south, and west of a particular site are viable targets to receive offspring. We use periodic boundary conditions. The scheme for motion is the same as for the cycle: at the beginning of an update step, an individual, chosen uniformly at random, swaps position with another individual chosen uniformly-at-random who is not more than $d$ lattice steps away. This process is repeated $J$ times before another DB update occurs.

The results of our simulations are shown in Fig. \ref{fig:DBprobplotLAT}. Even in the absence of motion, the fixation probability on a lattice is already less than the well-mixed case due to the effects of structure on the forward bias $\gamma$ when there are either $1$ or $N-1$ mutants \cite{kaveh:RSOS:2015}. As motion is introduced, we see further suppression of the fixation probability. We did not simulate fixation probabilities for BD updating on the lattice because the forward bias in this case is always $\gamma =r$, which is guaranteed by the fact that the lattice is regular \citep{lieberman:Nature:2005}. Therefore, motion does not change fixation probabilities under BD updating.

\begin{figure}
\centering
\includegraphics[width=0.8\columnwidth]{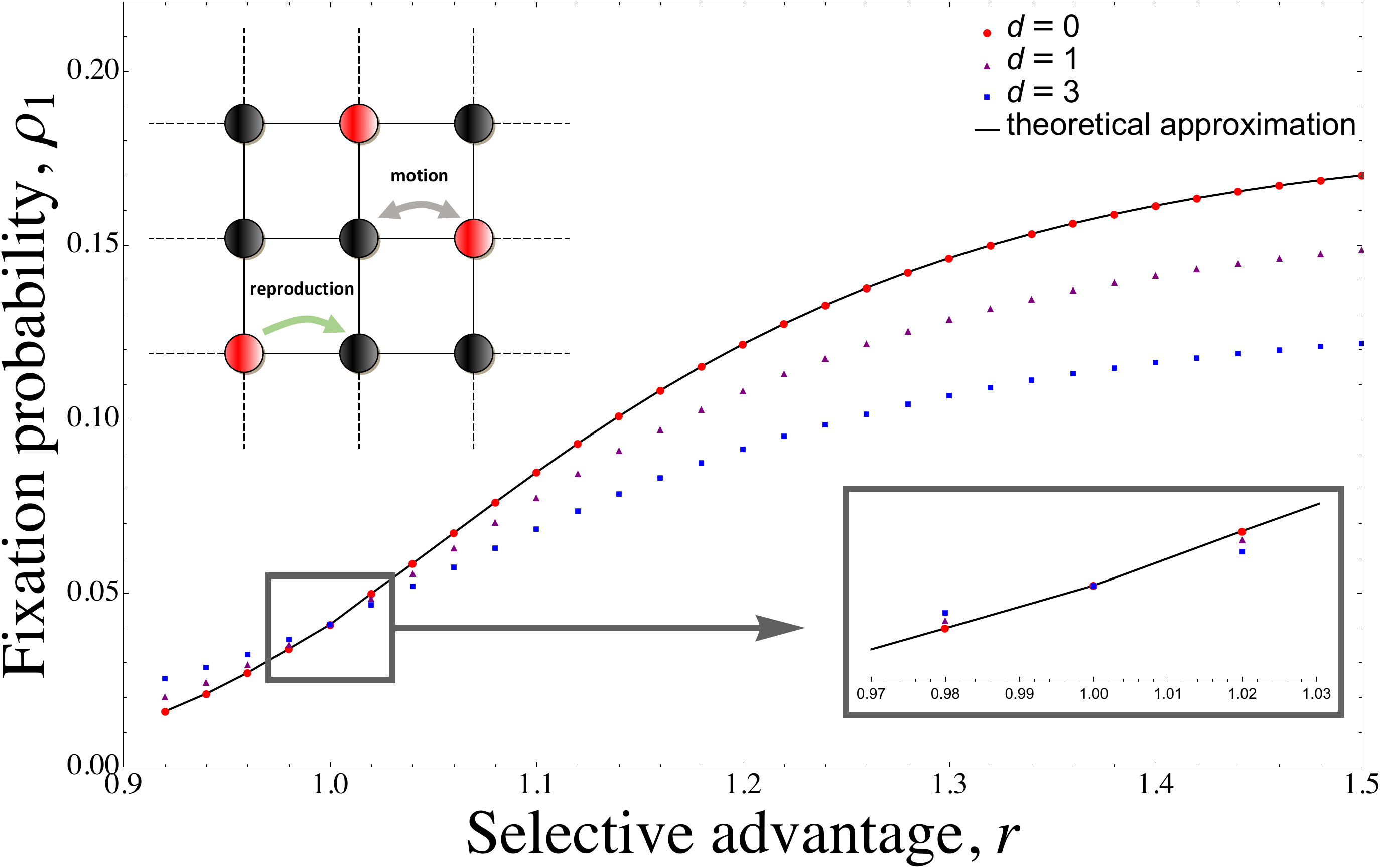}
\caption{\small Death-birth updating on a $5 \times 5$ square lattice with motion, schematized in top-left insert. The lattice is periodic to avoid boundary effects. Individuals are of either the resident type (black) or mutant type (red). The update step is comprised of two events. First, one individual is selected at random and swapped with an individual at most $d$ lattice steps away (gray arrow). Second, an individual is selected uniformly-at-random to die, and is replaced by a neighbor with probability proportional to fitness (green arrow). The curves shown have $d=0$, which is equivalent to no motion (red circles), $d=1$ (purple triangles), and $d=3$ (blue squares). These results are the average of $10^6$ realizations. The fixation probability on a lattice without motion is also given by a solid black line, drawn from an existing theoretical approximation known in the literature \citep[see][]{kaveh:RSOS:2015}, which agrees well with our simulation results. We find that motion suppresses natural selection on the two-dimensional lattice much like on the one-dimensional lattice. \label{fig:DBprobplotLAT}}
\end{figure}

Despite being able to calculate the forward bias $\gamma$ explicitly, our methods for proving suppression on the cycle do not extend to the two-dimensional lattice. Interestingly, there exist configurations of $m\geq3$ mutants for which $\gamma >r$; therefore, it cannot be that on the two-dimensional lattice motion uniformly decreases the forward bias, which was the key to the proof of suppression on the cycle. On the lattice, motion can transiently carry an advantageous ($r>1$) mutant population through a structure that temporarily increases their chances of replacing the native type. This property also suggests that type-dependent motion, where mutants can rearrange themselves into these beneficial structures, could be a valid strategy to amplify their selective advantage. Type dependent motion may also be important when the two species have different motilities --- a standard example is the invasion of cane toads in the South Pacific \citep{phillips:JEB:2010}.

\subsection{Weighted graphs}
So far, we have used graphs simply to indicate who is a neighbor of whom. We can further associate to each (directed) edge a weight, which can account for asymmetry in dispersal patterns \citep{lieberman:Nature:2005}. Suppose that $\Gamma_{ij}\geqslant 0$ indicates the weight of the edge from vertex $i$ to vertex $j$, and $\sum_{j=1}^{N}\Gamma_{ij}=1$ for each $i$. Under DB updating on this graph, an individual, $i$, is first selected for reproduction with probability proportional to relative fitness. The individual at vertex $j$ then dies and is replaced by the offspring of $i$ with probability $\Gamma_{ij}$. Between these updates, as before, motion can rearrange the types residing on the graph.

Remarkably, motion can either suppress or amplify selection on weighted graphs. It can be the case that one particular type of motion on a given weighted graph suppresses selection and another amplifies selection (see Fig. \ref{fig:weightedGraph}). It is also possible to construct a weighted graph on which a particular type of swapping motion both amplifies and suppresses selection: it increases a mutant's fixation probability relative to the static case for certain values of $r>1$ and decreases it for others (Fig. \ref{fig:notAmpSupp}). On a general weighted graph it may therefore be difficult to classify the effects of motion on selective differences between types, which is reminiscent of the difficulties that arise in classifying population structures themselves \citep{nowak:PTRSB:2009}.

\begin{figure}
\centering
\includegraphics[width=0.8\columnwidth]{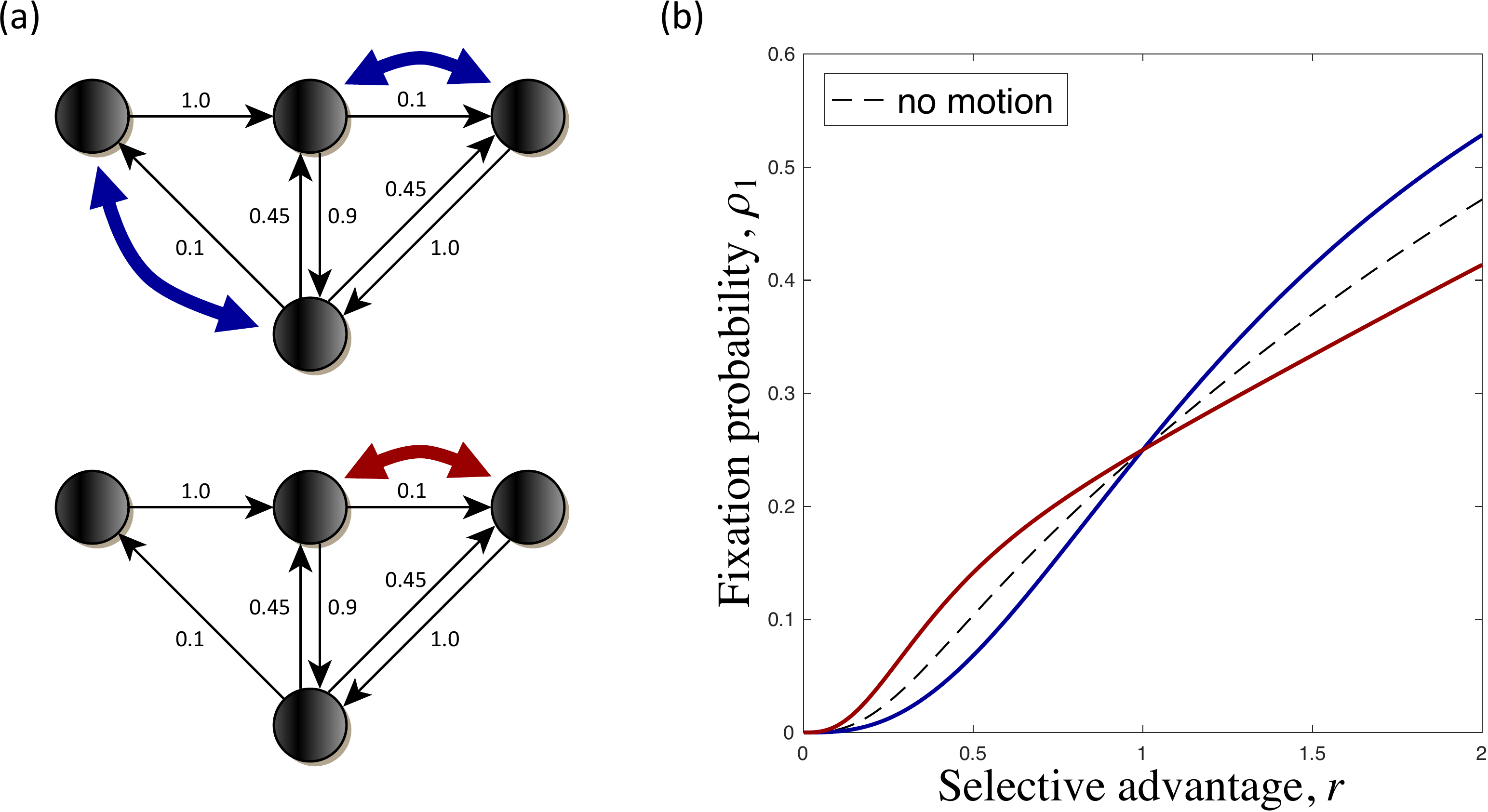}
\caption{\small Motion on a weighted graph. On general graphs, motion can exhibit behavior consistent with either amplification or suppression of selection. Panel (a) depicts a small weighted graph with two types of distinct motion, given in blue and red, respectively. For the blue motion, two swapping events occur simultaneously after each (birth-death) update: the individuals connected by blue arrows are swapped. While this motion involves multiple swaps, it is still an abstract shuffle. For the red motion, only one swapping event occurs after each reproduction event, and the individuals connected by the red arrow are swapped. The weights on the graph represent probabilistic dispersal patterns should an individual on the tail end of an edge be selected for reproduction. Relative to the process with no motion (dashed line), the motion in blue amplifies selection and the motion in red suppresses selection, shown in Panel (b). On the one- and two-dimensional lattices, in contrast, motion cannot amplify selection (even if more than one swap is allowed at each time step). Therefore, for more complicated dispersal patterns and weighted graphs, motion need not act uniformly as a suppressor as it does on unweighted one- and two-dimensional lattices.\label{fig:weightedGraph}}
\end{figure}
\begin{figure}
\centering
\includegraphics[scale=0.5]{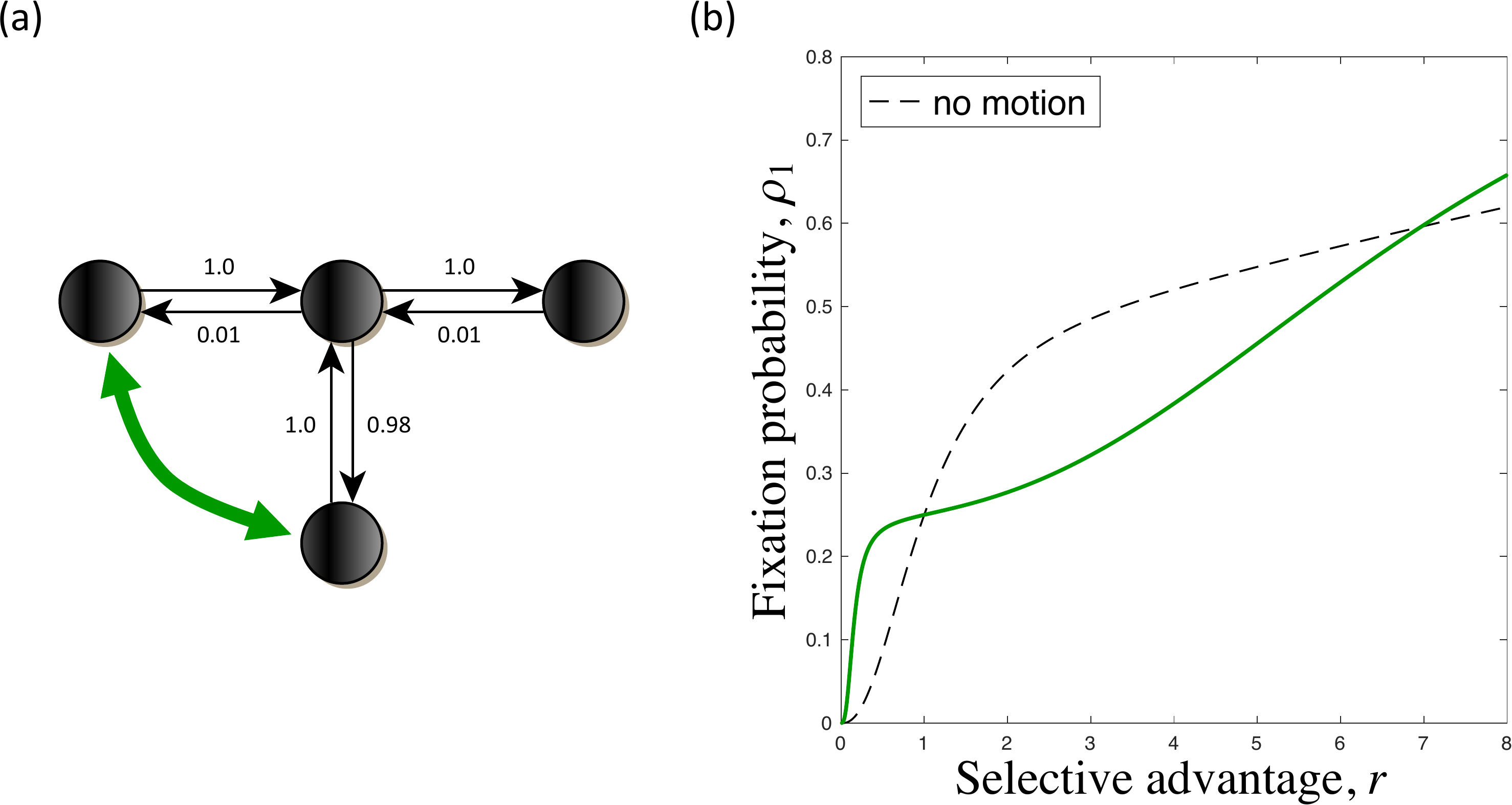}
\caption{\small Motion on a graph need not be a strict amplifier or a suppressor of selection. The swapping of types between each reproductive event is indicated by a green arrow in the graph depicted in Panel (a); dispersal probabilities for birth-death updating are indicated above each directed edge (black). Panel (b) shows the fixation probability of a randomly-placed mutant as a function of the mutant's relative fitness, $r$. Since this fixation probability coincides with the fixation probability when motion is absent for at least two (but still finitely many) distinct values of $r$, this motion on the graph is neither a strict suppressor nor an amplifier of selection, but plays both roles in different regions of $r$.\label{fig:notAmpSupp}}
\end{figure}

\subsection{Interpreting motion on graphs}
The study of motion in arbitrary graph-structured populations raises several issues related to how motion occurs in realistic populations. In our idealized case, individuals have physical limitations set forth by the vertices of a graph. A shuffle ensures that this population structure is preserved by motion. Following motion, all individuals fall somewhere on the original graph, and motion affects only the configuration of types within a fixed population structure.

Just as populations themselves need not have fixed size and structure, motion need not preserve the structure of a population. Even motion that does preserve population structure can take on abstract and complicated forms, which can be problematic if one is to use this framework to directly model motion within a natural population. Within the class of motion that does not depend on the individuals' types, there are two main subclasses: motion that depends on population structure and motion that does not.

Shuffles within a fixed population structure can depend on the population structure itself. For example, on the grid depicted in Fig. \ref{fig:motionOnCycle}(C), swaps occur between neighbors and not between more distant pairs of inhabitants. In a dense population, such as one composed of bacterial swarms \citep{darnton:BJ:2010,vanditmarsch:CR:2013} or tissue cells \citep{aktipis:NRC:2013}, the number of contacts between cells is roughly constant, and it is reasonable to consider a static, regular graph with swappings as cells are jumbled around by thermal or driven motion. Motion of this form, although not dependent on the types of the individual cells, is clearly coupled with population structure.

On the other hand, if motion does not depend on the population structure, then swaps like that of Fig. \ref{fig:notAmpSupp} are reasonable since they do not depend on the routes through which parents can propagate their offspring (i.e. links on the graph). While a swap of two distant individuals might seem contrived in this scenario, other shuffles resembling shifts or diffusions are more relevant--especially in large populations or those with periodic boundaries (such as a lattice). Structure-independent motion might arise from wind \citep{reich:CJFR:1994,yokomizo:TE:2009} or water flow \citep{perlekar:PRL:2010}.

Our focus has been on modeling motion on a fixed graph. As the density of individuals decreases or the rate of motion increases, however, it may be more realistic to consider a dynamic graph, in which individuals moving around on a background can have varying numbers of neighbors at different times. We touch upon this type of motion in the next section. At the same time, less is known about evolution in populations with dynamic structure, so even in the absence of motion further work must be done.

The cycle, while still an extremely simple population structure, is unique among evolutionary graphs with motion. The notion of shuffling on a static graph also captures more complicated motion when the underlying graph is a cycle. If every individual has exactly two neighbors and the population is rearranged by thermal excitation or fluid flow, for instance, then the requirement that everyone (in a connected population) has exactly two neighbors results once again in a cycle following the flow. As a result, motion such as a flow effectively induces a shuffle of the individual types within the cycle. This property is generally not true if one looks beyond the cycle--in fact, even for graphs in which all individuals have exactly three neighbors, a flow typically changes the topology of the graph. Therefore, our treatment of the cycle covers a variety of classes of motion beyond what the notion of shuffling could capture on more complicated graphs.

\section{Motion as dynamic topology}
We now consider evolutionary dynamics on a graph whose topology is not fixed. A general recipe for generating this class of motion is as follows: first, individuals are assigned initial positions on a manifold. Then, a graph is generated by essentially one of two methods; either a regular graph is generated, taking into account the metric on the underlying space (such that points closer together are more likely to share an edge), or a heterogeneous graph is generated by assigning an edge between each individual and all others within a metric ball of radius $R$ of that individual. The latter characterizes a random geometric graph \citep{penrose:OUP:2003}, and is our choice for generating topologies. Between Moran update steps, players are repositioned according to a map from the manifold to itself, and a new graph is formed according to the same rules. Thus, the map that moves the players effectively induces a new graph at every update step, which gives rise to an evolutionary process on a dynamic graph.

On a random geometric graph, whether the topology is static or dynamic, the fixation probability does not change when the update rule is birth-death. This property arises from the fact that for a sufficiently large random graph, it is known from the robust isothermal theorem \citep{adlam:SR:2014} that the fixation probability converges to that of the well-mixed model. We therefore focus our attention on the death-birth update rule.

In this initial study, we consider $N$ individuals distributed uniformly-at-random over the flat unit torus, with one individual being a mutant with relative fitness $r$. We choose a map from the torus to itself that has strong chaotic properties (see \citep{thiffeault:C:2003} and \textbf{Supporting Information}), to ensure that each individual has frequent opportunities to share an edge with any other individual in the population. With the map fixed, the only free parameter is the radius, $R$, of the ball that induces the graph at each update step. The results of our simulations are shown in Fig. \ref{fig:CATmap} for $N=49$ and various values of $R$. Despite the fact that the topology is completely new at each update step, suppression of selection due to motion is still seen when compared to a well-mixed population.

\begin{figure}
\centering
\includegraphics[width=0.8\columnwidth]{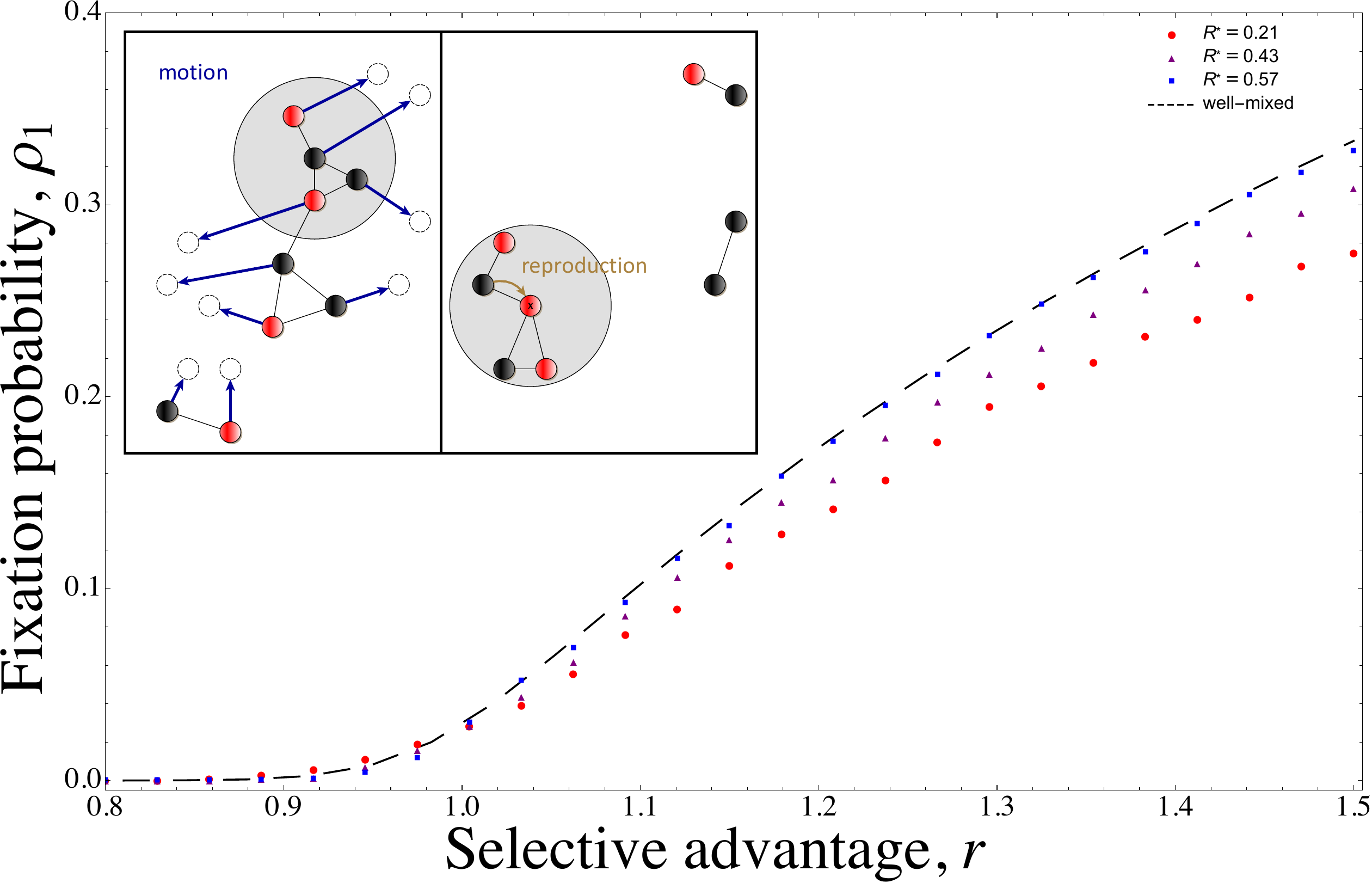}
\caption{Selection on a dynamic graph generated by a discrete map $f:\mathbb{R}^{2}\rightarrow\mathbb{R}^{2}$, schematized in top-left insert. At the end of a time step, the population is structured on a graph given by the rule that an individual $i$ has an edge with individual $j$ if the Euclidean distance between $i$ and $j$ is less than some radius, $R$, illustrated here as a grey disk. At the beginning of the next time step, this graph is then changed by moving the individuals according to the map, such that the position of individual $i$ becomes $f\left(i\right)$ (green arrows), and again applying the edge creation rule. The update rule ends with reproduction. The plots show the fixation probability as a function of the selective advantage of a single initial mutant, for a fixed map (see \textbf{Supporting Information}), $N=49$, and three values of $R$: $R=0.21$ (red circles), $R=0.43$ (purple triangles), $R=0.57$ (blue squares). The well-mixed result is given by the dashed black line.\label{fig:CATmap}}
\end{figure}

Extending this system to more diverse spaces and maps could be an intriguing avenue of study. In particular, it might provide a proper framework for understanding evolutionary game theory in a motion-dependent context. The size of the ball, $R$, is intimately related to both (i) the maximum distance that one player can signal their intention to another player and (ii) the maximum distance that a public good can diffuse from the player who generated it before it is advected away by an ambient flow or moving background. In this context, motion could be a determining factor in whether certain social behaviors or game strategies can flourish in a population.

\section{Discussion}
Population structures that amplify or suppress selection have received much attention in recent years \citep{adlam:PRSA:2015,jamiesonlane:JTB:2015,galanis:amplifiers:2015,giakkoupis:amplifiers:2016,goldberg:amplifiers:2016}. While a population structure is typically classified as an amplifier or suppressor by comparison to the unstructured population (the complete graph), here we classify a process with motion by comparing it to the same population in the absence of motion. In principle, one could also compare a population with motion to an unstructured population, but this comparison would make it difficult to disentangle the effect of motion from that of the spatial structure. This raises a critical point about the definition of natural selection. As we have seen, in the presence of motion, increasing growth rate does not necessarily increase selective advantage over other phenotypes. Selection acts on some combination of motility and growth rate. Understanding if there is a transformation under which a population with a certain motility and growth rate can be represented by a static population with a different growth rate, such that selection is only acting on one quantity, may be a useful future inquiry. 

Motion among individuals is a modeling component with meaningful biological motivation \citep{yokomizo:TE:2009,smaldino:TPB:2012,mcmanus:JEB:2012}. We find that motion on a static structure can act as amplifier or suppressor of natural selection relative to the same structure in the absence of motion. For the one-dimensional lattice, the cycle, we prove that any motion, however complicated, suppresses selection. This motion can even include dynamic rearrangement of the individuals in a continuous space, since the cycle is preserved under this motion as long as each individual finds two unique neighbors at every reproductive step.

Birth-death (BD) and death-birth (DB) are standard update rules, whose behavior can differ \citep{ohtsuki:Nature:2006,ohtsuki:JTB:2006}. Our study applies to both BD and DB updating, and it also extends to a stochastic mixture of BD and DB \citep{zukewich:PLoSONE:2013}. This mixture is a continuous interpolation between pure BD and pure DB updating and is biologically relevant since it accounts for birth and death events that do not occur in a fixed order. That any non-trivial motion suppresses selection on the cycle is a robust result that holds irrespective of whether BD updates are more likely than DB updates.

Apart from its effect on fixation probabilities, motion shortens the absorption time and the conditional fixation time. This acceleration arises from the fact that motion leads to fragmentation of clusters. On the cycle without motion, regardless of whether the update rule is BD or DB, the population can change only when reproduction events take place on the two boundaries that separate the two types. Motion engenders more boundaries between clusters, speeding up the dynamics. 

In any evolutionary graph, the number of possible events that increase or decrease the number of mutants is equal to the number of edges between mutant-occupied and resident-occupied vertices. The involved vertices are called the ``cut-set" \citep{lieberman:Nature:2005}. The class of $d$-neighbor swaps in general increase the size of the cut-set, increasing the frequency of events in which the mutant population grows or shrinks. For a given structure a motion can be defined that is ``evolutionarily optimal'' in one of two senses: an optimal motion can either maximize the fixation probability of an advantageous mutant, or it can minimize the conditional fixation time. These two optimal motions need not be the same. As we have seen for DB updating, speeding up the dynamics involves creating more clusters of the mutant type, which reduces the selective difference. Therefore, it could be difficult to find an optimal dispersal strategy that decreases the time for a new advantageous type to fix without removing its advantage.

All of this insight extends to dynamic graphs as well. Having a dynamic graph topology ensures that the composition of a local neighborhood is constantly changing, so that meaningful interactions take place much more frequently, and individuals of a given type are much more likely to encounter individuals of a different type than if the population structure were static. In this context, it is not surprising that the suppression of selection seen on the fixed one- and two-dimensional lattices extend to our example of a dynamic graph. 

Increasing the amount of motion between consecutive reproductive events on the lattice or on the cycle does not lead to the so-called ``well-mixed'' population, which is represented by a complete graph. As we have seen, the fixation probabilities in these settings are very different for DB updating. Motion on weighted and/or heterogeneous-degree graphs can lead to more nuanced results than what we have observed on the cycle and the lattice. Here, motion can act as amplifier or suppressor, and in some cases whether it is one or the other depend on the selective difference between the types.

Additionally, motion could be type-dependent. Motion could act differently on different types of individuals, which has been considered in previous work \citep{hamilton:Nature:1977}. Our formal proof for the cycle covers type-dependent motion and shows that it can only suppress selection. However, this result does not extend beyond the cycle, and it is easy to construct counterexamples. For instance, on any heterogeneous-degree graph, mutants that control their motion could gain an immediate advantage by moving themselves to particularly beneficial vertices with large numbers of neighbors. ``Intelligent'' types can perform motion that takes advantage of local benefits. Therefore ``motion with intent'' might be an interesting topic for future work.

\section*{Acknowledgments}
This work was supported by the Office of Naval Research (grant N00014-16-1-2914) and the John Templeton Foundation. The Program for Evolutionary Dynamics is supported in part by a gift from B Wu and Eric Larson. The authors thank Ben Adlam and Kamran Kaveh for helpful conversations.

\newpage

\setcounter{section}{0}
\setcounter{equation}{0}
\setcounter{figure}{0}
\renewcommand{\thesection}{SI.\arabic{section}}
\renewcommand{\theequation}{SI.\arabic{equation}}

\begin{center}
\textbf{Supporting Information}
\end{center}

\section{Motion suppresses selection on the cycle}
Let $M$ be the transition matrix for an evolutionary process on $S^{N}$, where $S=\left\{A,B\right\}$. For any state, $s$, let $\left| s\right|$ denote the number of mutants in $s$. We assume that for each $s\in S^{N}$, $M_{ss'}=0$ whenever $\left| s'\right| -\left| s\right|\neq 0,\pm 1$ (which holds for all of the processes we consider in the main text). This process defines a Markov chain on $S^{N}$, which we denote by $\left\{X_{n}\right\}_{n\geqslant 0}$, with the property that at most one mutant is added to or subtracted from the population at each step. For any non-absorbing state, $s\in S^{N}$, we define the forward bias by
\begin{linenomath}
\begin{subequations}
\begin{align}
p_{s;+} &:= \sum_{\substack{s'\in S^{N} \\ \left| s'\right| = \left| s\right| +1}} M_{ss'} ; \\
p_{s;-} &:= \sum_{\substack{s'\in S^{N} \\ \left| s'\right| = \left| s\right| -1}} M_{ss'} ; \\
\gamma_{s} &:= \frac{p_{s;+}}{p_{s;-}} .
\end{align}
\end{subequations}
\end{linenomath}

Let $\rho_{s}$ denote the probability of ending up in the all-mutant absorbing state when starting from state $s$. Recall that motion on the graph is modeled as a stochastic shuffle, $\mu\in\Delta\left(\mathfrak{S}_{N}\right)$, at each time step. (In general, $\mu$ can change at each update step.) If $\mu$ and $\nu$ are stochastic shuffles, then the sequence of $\mu$ followed by $\nu$, $\nu\circ\mu$, is again a shuffle with
\begin{linenomath}
\begin{align}
\nu\circ\mu\left(\pi\right) &= \sum_{\substack{\sigma,\tau\in\Delta\left(\mathfrak{S}_{N}\right) \\ \tau\circ\sigma =\pi}} \mu\left(\sigma\right)\nu\left(\tau\right) .
\end{align}
\end{linenomath}
In words, the probability of realizing $\pi$ as an effective two-shuffle outcome is the probability of all paths (sequences) of shuffles of length two that lead to $\pi$. This property extends to sequences of shuffles of any length. Thus, even motion that is iterated many times between each time step can always be captured by a stochastic shuffle. With this definition in place, the main technical lemma we need in order to show that motion suppresses selection is the following:
\begin{lemma}\label{lem:mainLemma}
Suppose that $P$ and $Q$ are Markov chains on $\left\{0,1,\dots ,m\right\}$ with the property that $\gamma_{\left| s\right|}^{P}\leqslant\gamma_{s}\leqslant\gamma_{\left| s\right|}^{Q}$ for each $s\in S^{N}$. Then, for every non-absorbing initial state, $s\in S^{N}$, we have $\rho_{\left| s\right|}^{P}\leqslant\rho_{s}\leqslant\rho_{\left| s\right|}^{Q}$.
\end{lemma}
\begin{proof}
By the recurrence relation for fixation probabilities,
\begin{linenomath}
\begin{align}
\sum_{s'\neq s} M_{ss'} \rho_{s} &= \sum_{s'\neq s} M_{ss'} \rho_{s'} .
\end{align}
\end{linenomath}
Therefore, whenever $s$ is not an absorbing state and $s'\neq s$, replacing $M_{ss'}$ by $\frac{1}{1-M_{ss}}M_{ss'}$ does not change these fixation probabilities; we may assume then that the probability of staying put in a given state is $0$ in any Markov chain for which we care about fixation probabilities. In other words, we have
\begin{linenomath}
\begin{subequations}
\begin{align}
p_{s;+} &= \frac{\gamma_{s}}{1+\gamma_{s}} ; \\
p_{s;-} &= \frac{1}{1+\gamma_{s}} .
\end{align}
\end{subequations}
\end{linenomath}
Consider the function $f_{P}:\left\{0,1,\dots ,N\right\}\rightarrow\mathbb{R}$ defined by
\begin{linenomath}
\begin{align}
f_{P}\left(m\right) &= \begin{cases}0 & m=0 , \\ 1+\frac{1}{\gamma_{1}^{P}}+\frac{1}{\gamma_{1}^{P}\gamma_{2}^{P}}+\cdots +\frac{1}{\gamma_{1}^{P}\cdots\gamma_{m-1}^{P}} & m>0 .\end{cases}
\end{align}
\end{linenomath}
Since $f_{P}\left(m+1\right) =f_{P}\left(m\right) +\frac{1}{\gamma_{1}^{P}\cdots\gamma_{m}^{P}}$ for each $m<N-1$, we have
\begin{linenomath}
\begin{align}
\mathbf{E}\left[ f_{P}\left(X_{n+1}^{P}\right)\ |\ X_{n}^{P}=m \right] &= \frac{\gamma_{m}^{P}}{1+\gamma_{m}^{P}} f_{P}\left(m+1\right) + \frac{1}{1+\gamma_{m}^{P}} f_{P}\left(m-1\right) \nonumber \\
&= f_{P}\left(m\right) .
\end{align}
\end{linenomath}
Therefore, by Doob's optional stopping theorem, we see that $\rho_{m}^{P}=f_{P}\left(m\right) /f_{P}\left(N\right)$. Now, since
\begin{linenomath}
\begin{align}
\mathbf{E}\left[ f_{P}\left(\left| X_{n+1}\right|\right)\ |\ X_{n}=s \right] &= \frac{\gamma_{s}}{1+\gamma_{s}} f_{P}\left(\left| s\right| +1\right) + \frac{1}{1+\gamma_{s}} f_{P}\left(\left| s\right| -1\right) \nonumber \\
&\geqslant \frac{\gamma_{\left| s\right|}^{P}}{1+\gamma_{\left| s\right|}^{P}} f_{P}\left(\left| s\right| +1\right) + \frac{1}{1+\gamma_{\left| s\right|}^{P}} f_{P}\left(\left| s\right| -1\right) \nonumber \\
&= f_{P}\left(\left| s\right|\right) ,
\end{align}
\end{linenomath}
it follows once again from Doob's optional stopping theorem that $\rho_{s}\geqslant f_{P}\left(\left| s\right|\right) /f_{P}\left(N\right) =\rho_{\left| s\right|}^{P}$. An analogous argument with the inequalities reversed gives $\rho_{s}\leqslant\rho_{\left| s\right|}^{Q}$, which completes the proof.
\end{proof}

\subsection{Mixed BD and DB updating}
Consider a mixed update rule in which, at each time step, there is a DB update with probability $\delta$ and a BD update with probability $1-\delta$. For each state, $s$, we have $p_{s;+}^{\delta}=\left(1-\delta\right) p_{s;+}^{\textrm{BD}}+\delta p_{s;+}^{\textrm{DB}}$ and $p_{s;-}^{\delta}=\left(1-\delta\right) p_{s;-}^{\textrm{BD}}+\delta p_{s;-}^{\textrm{DB}}$, so the forward bias is
\begin{linenomath}
\begin{align}
\gamma_{s}^{\delta} &= \frac{\left(1-\delta\right) p_{s;+}^{\textrm{BD}}+\delta p_{s;+}^{\textrm{DB}}}{\left(1-\delta\right) p_{s;-}^{\textrm{BD}}+\delta p_{s;-}^{\textrm{DB}}} .
\end{align}
\end{linenomath}
If $1<\left| s\right| <N-1$ and $s'$ is chosen so that $\left| s'\right| =\left| s\right|$ but all mutants in $s'$ are in a single cluster, then $x_{A}\left(s'\right) =x_{B}\left(s'\right) =0$ and $y_{A}\left(s'\right) =y_{B}\left(s'\right) =1$; a simple calculation then gives $\gamma_{s'}^{\delta}=r$. If $r\geqslant 1$, then
\begin{linenomath}
\begin{align}
\gamma_{s}^{\delta} &= \frac{\left(1-\delta\right)\left(\frac{r\left(x_{A}+y_{A}\right)}{m r + N-m}\right) + \delta\left(\frac{x_{B} + \frac{2r}{1+r} y_{B}}{N-2x_{A}-2y_{A}}\right)}{\left(1-\delta\right)\left(\frac{x_{A}+y_{A}}{m r + N-m}\right) + \delta\left(\frac{x_{A} + \frac{2}{1+r} y_{A}}{N-2x_{A}-2y_{A}}\right)} \nonumber \\
&\leqslant \frac{\left(1-\delta\right)\left(\frac{r\left(x_{A}+y_{A}\right)}{m r + N-m}\right) + \delta\left(\frac{2r}{1+r}\right)\left(\frac{x_{B} + y_{B}}{N-2x_{A}-2y_{A}}\right)}{\left(1-\delta\right)\left(\frac{x_{A}+y_{A}}{m r + N-m}\right) + \delta\left(\frac{2}{1+r}\right)\left(\frac{x_{A} + y_{A}}{N-2x_{A}-2y_{A}}\right)} \nonumber \\
&= r
\end{align}
\end{linenomath}
since $x_{A}+y_{A}=x_{B}+y_{B}$ on the cycle. Thus, $\gamma_{s}^{\delta}\leqslant\gamma_{s'}^{\delta}$ when $r\geqslant 1$. Similarly, we see that $\gamma_{s}^{\delta}\geqslant 1$ when $r\geqslant 1$. When $r\leqslant 1$, we get the opposite inequalities, namely $r=\gamma_{s'}^{\delta}\leqslant\gamma_{s}^{\delta}\leqslant 1$. Therefore, for any $r>0$,
\begin{linenomath}
\begin{align}\label{sieq:gammaInequality}
\min\left\{\gamma_{s'}^{\delta},1\right\} \leqslant \gamma_{s}^{\delta} \leqslant \max\left\{\gamma_{s'}^{\delta},1\right\} .
\end{align}
\end{linenomath}
The biases for the two remaining cases, $\left| s\right| =1$ and $\left| s\right| =N-1$, are
\begin{linenomath}
\begin{align}
\gamma_{s}^{\delta} &= \begin{cases}\frac{\left(1-\delta\right)\left(\frac{r}{m r + N-m}\right) + \delta\left(\frac{\frac{2r}{1+r}}{N-2}\right)}{\left(1-\delta\right)\left(\frac{1}{m r + N-m}\right) + \delta\left(\frac{1}{N-2}\right)} & \left| s\right| =1 , \\ \frac{\left(1-\delta\right)\left(\frac{r}{m r + N-m}\right) + \delta\left(\frac{1}{N-2}\right)}{\left(1-\delta\right)\left(\frac{1}{m r + N-m}\right) + \delta\left(\frac{\frac{2}{1+r}}{N-2}\right)} & \left| s\right| =N-1 .\end{cases}
\end{align}
\end{linenomath}
While these need not be equal to $r$, Eq. (\ref{sieq:gammaInequality}) still holds since $s=s'$ in these two cases.

It follows at once from Eq. (\ref{sieq:gammaInequality}) and Lemma \ref{lem:mainLemma} that motion on the cycle suppresses selection under mixed BD and DB updating since any such (non-trivial) motion can disrupt clusters. In particular, this suppression holds for $\delta =0$ (pure BD updating) and $\delta =1$ (pure DB updating).

\section{Simulations on one- and two-dimensional lattices}
\subsection{BD updating}
We simulated BD updating on the cycle $10^6$ times for $N=25$, thirty values of $r$ evenly distributed on $[0.5, 1.5]$, and different values of $d$. At first, we constrain ourselves to one shuffle per update, $J=1$. As predicted, the fixation probability is identical to the well-mixed result, but the times to fixation and absorption can be different. Generally speaking, when the motion is independent of the state and position of the individuals, the time to absorption or fixation is longer than in the well-mixed case. However, if $\mu$ depends on the state rather than being drawn from the uniform distribution, we can provide an algorithm for the shuffles that provide the fastest dynamics, even faster than the well-mixed case. The fastest realizations for BD updating are those that minimize $p^0$ at every update. It is easily seen that the desired $\mu(s)$ maximizes $c$ for all $s$, that is, creates as many clusters as possible. This can easily be accomplished by an algorithm that ensures no two mutants are adjacent when $m\leq N/2$, and no two resident types are adjacent when $m>N/2$. We plot $p^0$ versus $m$ for this strategy and the well-mixed case in Fig. \ref{fig:fasterThanWM}. Because the biases are the same in the two cases and $p^0$ for the well-mixed case dominates $p^0$ for the algorithm, we actually expect the algorithm on the cluster to be faster than dynamics in the unstructured, well-mixed case, (see Fig. \ref{fig:ALGORITHM}). 

\subsection{DB updating}
Similarly, we simulated DB updating for the same parameters and iterations described in the previous section. The behavior of the conditional fixation time is shown in the first panel of Fig. \ref{fig:DBprobplot}. Quantitative differences arise due to the fact that the forward bias is different under DB and BD updating when motion is present, as described in the text. This difference in bias implies that the probability that an invading mutant will fix is different under DB updating when compared to the well-mixed and BD cases. The probability that the mutant will fix is plotted in the second panel of Fig. \ref{fig:DBprobplot}. We find that in the presence of random swapping, the forces of selection are inhibited; that is, the fixation probability with motion is dominated by the fixation probability without motion.

In considering the lattice, all of the properties and parameters described in this and the previous section hold. When individuals are chosen to swap, the distance $d$ is measured in units of lattice steps, or in the $\mathbf{L}^1$ (taxicab) norm. The probability that the mutant will fix is plotted in Fig. \ref{fig:DBprobplotLAT}.

\subsection{Multiple shuffles per update}
When $J$ is large, it is more fundamental to examine the role of various quantities as the percentage of swapped sites, $J/N$, is varied; if $d \ll N$, it can take a very large number of swaps to accomplish the same relative mixing as a small number of swaps when $d = \mathcal{O}(N)$. We fix values of $d$ and average over ten values of $N \in [25, 50]$ in Fig. \ref{fig:DBringJplot}, which examines the dependence of the fixation probability and conditional fixation time under DB updating on the amount of random swapping and the distance. As discussed in the text, the nature of the swapping implies that there is a critical value of $J/N$ beyond which there is no further effect in the average on fixation probabilities and times; this corresponds to the percolation limit. Depending on the value of $d$, we estimate that between $10$ and $20$ percent of individuals in our population must be swapped before this completely-randomized condition is met, leading to the fastest dynamics and lowest fixation probabilities.

\section{Weighted graphs}
We now illustrate different types of motion on weighted graphs. For simplicity, we assume that the population is updated according to a BD rule; similar behavior can be observed DB updating. Suppose that $\Gamma$ is a weighted graph with $N$ vertices. To each pair of vertices, $i$ and $j$, $\Gamma$ associates a weight, $\Gamma_{ij}\geqslant 0$. BD updating on $\Gamma$ (without motion) is defined as follows \citep{lieberman:Nature:2005}: First, an individual is chosen for reproduction with probability proportional to (relative) fitness. If $i$ is chosen to reproduce, then the offspring of $i$ replaces $j$ (who dies) with probability proportional to $\Gamma_{ij}$. Without a loss of generality, we may assume that $\sum_{j=1}^{N}\Gamma_{ij}=1$ for each $i$, meaning $\Gamma_{ij}$ is the probability of dispersal to location $j$ from location $i$.

Let $M$ be the transition matrix on $S^{N}$ induced by this process, where $S=\left\{A,B\right\}$. If $\mu\in\Delta\left(\mathfrak{S}_{N}\right)$ is the shuffle used before each BD update, then the transition matrix for BD updating on $\Gamma$ with motion induced by $\mu$ is defined by
\begin{linenomath}
\begin{align}
M_{ss'}^{\mu} &= \sum_{\pi\in\mathfrak{S}_{N}} \mu\left(\pi\right) M_{\pi\left(s\right) s'} .
\end{align}
\end{linenomath}
Let $\rho_{s}$ denote the probability of ending in the all-$A$ absorbing state when starting in state $s$. This fixation probability satisfies the recurrence relation,
\begin{linenomath}
\begin{align}
\rho_{s} &= \sum_{s'\in S^{N}} M_{ss'}^{\mu} \rho_{s'} ,
\end{align}
\end{linenomath}
with the boundary conditions $\rho_{\left(A,A,\dots ,A\right)}=1$ and $\rho_{\left(B,B,\dots ,B\right)}=0$. For small $N$, one can solve for these fixation probabilities directly. Doing so gives Figs. \ref{fig:weightedGraph} and \ref{fig:notAmpSupp} (where, by ``fixation probability," we mean the average fixation probability over all starting positions of the mutant; see Eq. (\ref{eq:averageFP})).

\section{Dynamic graphs}
We simulated the Moran process on a dynamic geometric graph induced by metric balls of radius $R$ on $N$ players, with the particular values $N=49$ and $R=1.5/7, \ 3/7, \ 4.5/7$. The values were picked based on conforming to a ``thermodynamic'' scaling $R=R^{*}/\sqrt{N}$ and approximating the values of $R$ such that the limit of a very dense (not quite complete) network was realized ($R=4.5/7$), as well as a much sparser network ($R=1.5/7$). This value was chosen so that the chance of a player being completely isolated (sharing no edges) was small. This probability cannot be eliminated on a finite dynamical graph, so in our simulations we have stipulated that any realization that passes through a transient state with an isolated player is discarded to avoid the possible bias induced by ``sparing'' death to an individual who cannot be replaced by any neighbors. We gleaned our results from $25,000$ realizations without isolated players for $25$ values of $r\in\left[0.8,1.5\right]$. 

The ambient space that was used was the flat unit torus, that is, $\left[0,1\right]^{2}$ with edges identified. In principle, any automorphism of the torus could be used to rearrange the players, and the study of the effect of different classes of automorphisms would be very interesting. As an initial study, however, we chose a very simple map with chaotic mixing properties:
\begin{linenomath}
\begin{align}
\begin{pmatrix}
x_{t+1} \\
y_{t+1}
\end{pmatrix} =
\begin{pmatrix}
2 & 1 \\
1 & 1
\end{pmatrix} 
\begin{pmatrix}
x_{t} \\
y_{t}
\end{pmatrix}
+ \varepsilon \begin{pmatrix}
\sin\left(2 \pi x_{t}\right) \\
\sin\left(2 \pi x_{t}\right) 
\end{pmatrix}
\end{align}
\end{linenomath}
that is, V.I. Arnold's CAT (continuous automorphism of the torus) map with a small perturbation term ($\varepsilon\ll 1$) to avoid the possibility of a stable fixed point for some initial conditions. The perturbation has no significant effect on the dynamics, which can be checked by omitting it and simulating the dynamics, taking great care to avoid the possibility of a fixed point ``trapping'' all the individuals. The map always has a positive Lyapunov exponent equal to the golden ratio, $\phi$, and is therefore chaotic for all initial conditions. There are no complicated regions of phase space such as elliptic islands that would be expected in genuine continuous fluid flows, so the map is in some sense very simple among the set of chaotic toral automorphisms. Its ergodic properties ensure that if the simulation is run forever (beyond the expected time to fixation of either species), each individual meets each other individual infinitely often. Such a map ensures that we are examining a very dynamic structure, with minimal possible correlation between edges between the adjacency matrix of two consecutive time steps. 

The results are shown in Fig. \ref{fig:CATmap}. As the connectivity of the graph decreases (decreasing $R$) from a complete graph, we observe suppression of selection due to the map; for sufficiently large $N$ a random geometric graph is isothermal, so in the absence of the map the fixation probability is close to that of a well-mixed population.

\newpage

\begin{center}
\textbf{Supporting Figures}
\end{center}

\setcounter{figure}{0}
\renewcommand{\thefigure}{SI.\arabic{figure}}

\begin{figure}[ht]
\centering
\includegraphics[width=0.8\columnwidth]{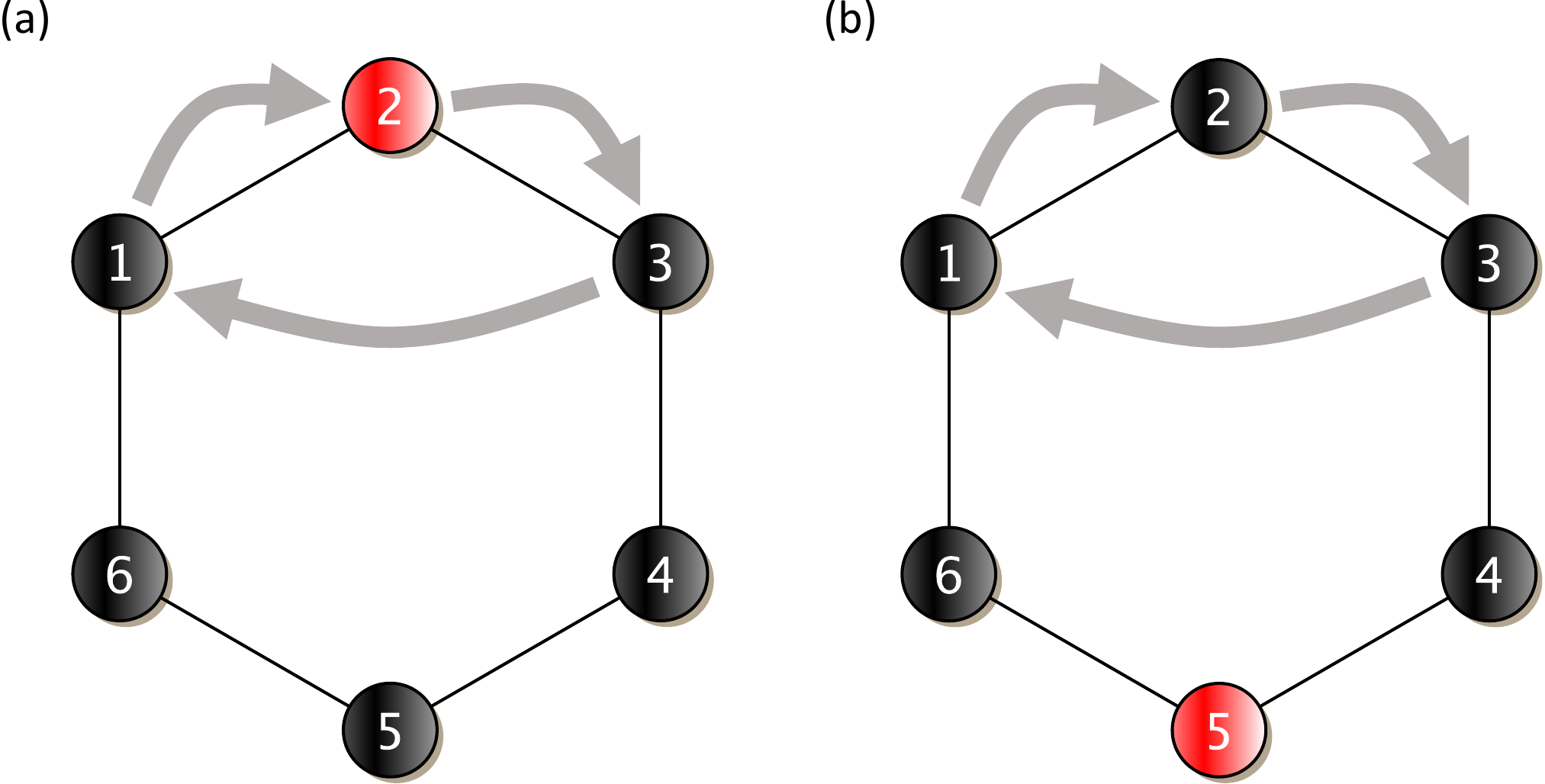}
\caption{\small Heterogeneity induced by motion. In the absence of motion, the fixation probability of a mutant at vertex $2$ is the same as that of a mutant at vertex $5$ since the cycle is homogeneous. Once motion is introduced, these fixation probabilities need not coincide. As an example, we can consider deterministic motion in which individual $2$ acquires the type of individual $1$, $3$ acquires the type of $2$, and $1$ acquires the type of $3$. As an element of the symmetric group, $\mathfrak{S}_{6}$, this motion is represented by the $3$-cycle $\left(132\right)$. For death-birth (DB) updating in which the mutant has fitness $r=2$ relative to the resident, the fixation probability is $\approx 0.3351$ in (a) and $\approx 0.3703$ in (b), both rounded to four digits beyond the decimal point. Therefore, even on a homogeneous graph, motion can introduce heterogeneity.\label{fig:heterogeneousMotion}}
\end{figure}
\begin{figure}[ht]
\centering
\includegraphics[width=0.8\columnwidth]{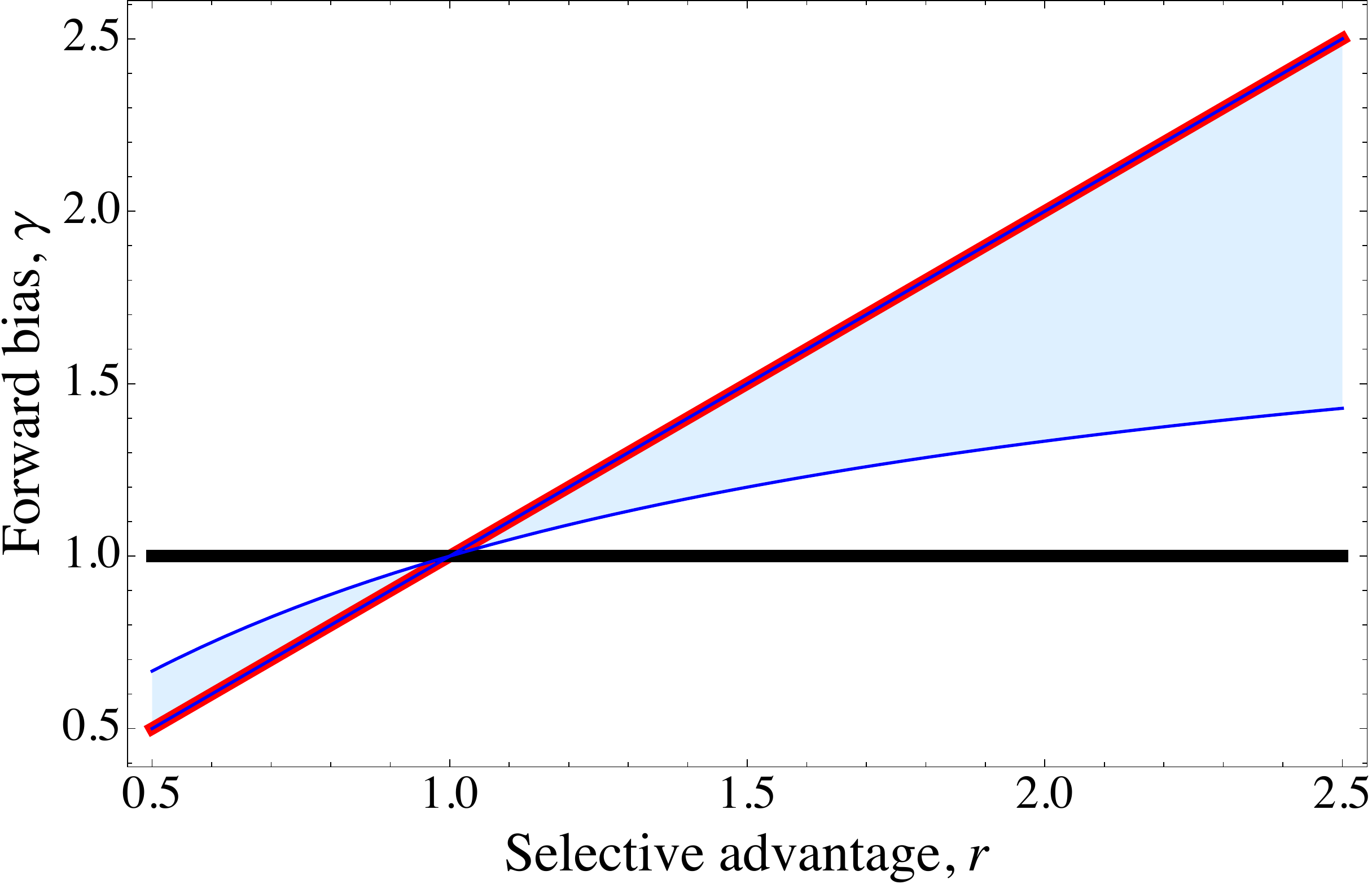}
\caption{\small Forward bias $\gamma$ for DB updating is decreased by motion. The example is that of five mutants arranged in different configurations on any cycle with $N \geq 15$. Thick black and red lines show the neutral limit $\gamma=1$ and the well-mixed limit $\gamma=r$, respectively. The blue lines give the maximum and minimum possible $\gamma$, realized when $x_B=0; \ y_B=1$ (a single cluster of mutants) and $x_B=5; \ y_B=0$ (five mutants isolated amongst large clusters of resident type), respectively. For both these extremes, $x_A=0$, indicating that the resident types are always in larger clusters. The shaded light blue region demarcates some other possible values of $\gamma$ that can be realized when the five mutants are dispersed in other configurations. \label{fig:gamCycPlot}}
\end{figure}
\begin{figure}[ht]
\includegraphics[width=0.8\columnwidth]{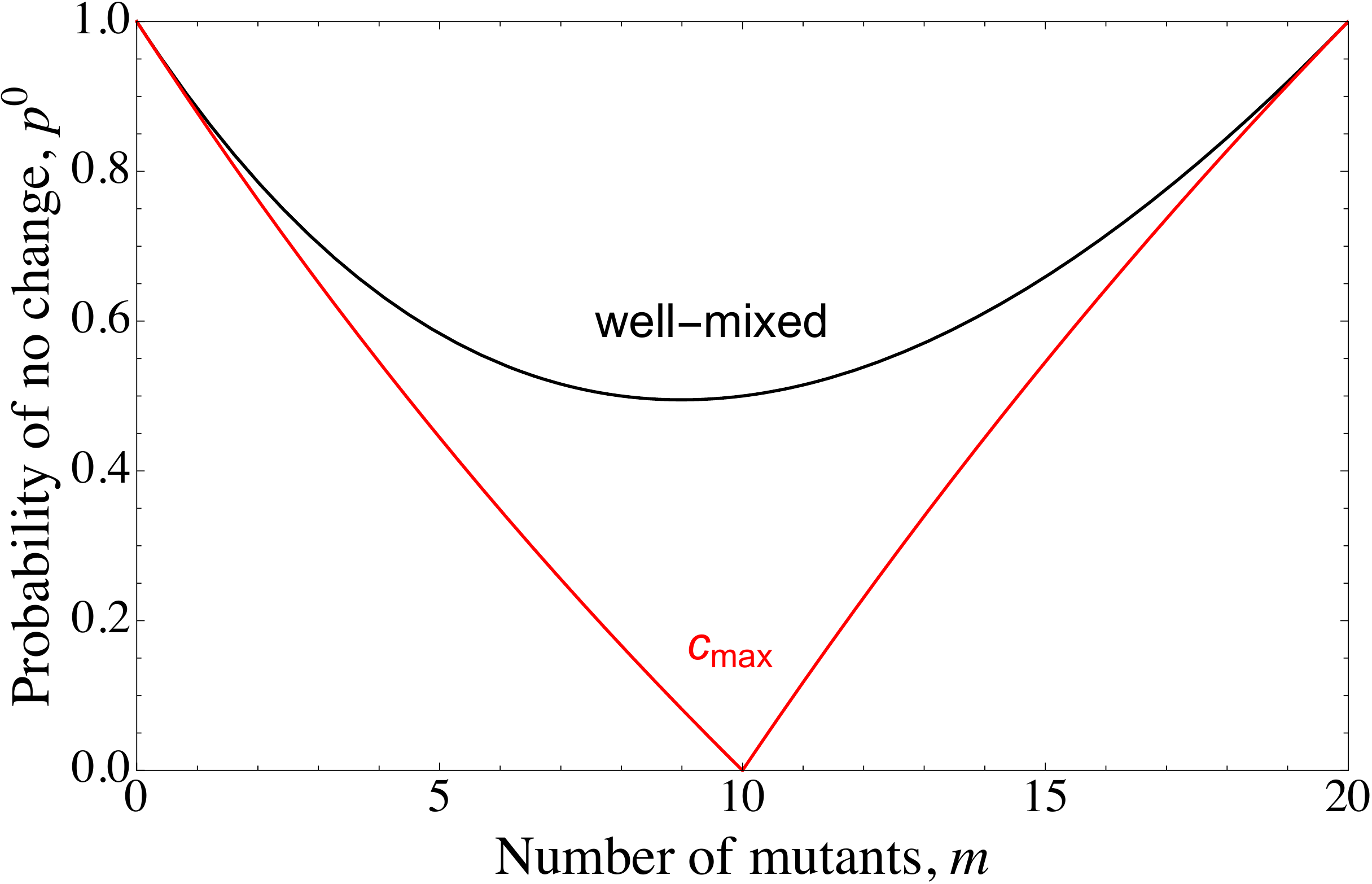}
\caption{\small Motion increases the frequency of birth and death events that change the abundance of types. The probability to neither create nor absorb a mutant under BD updating, $p^0$, is plotted versus $m$, for $N=20$. The black curve demonstrates the well-mixed case, whereas the red curve demonstrates the fastest possible dynamics, given by an algorithm that maximizes $c$ for every value of $m$.\label{fig:fasterThanWM}}
\end{figure}
\begin{figure}[ht]
\includegraphics[width=0.8\columnwidth]{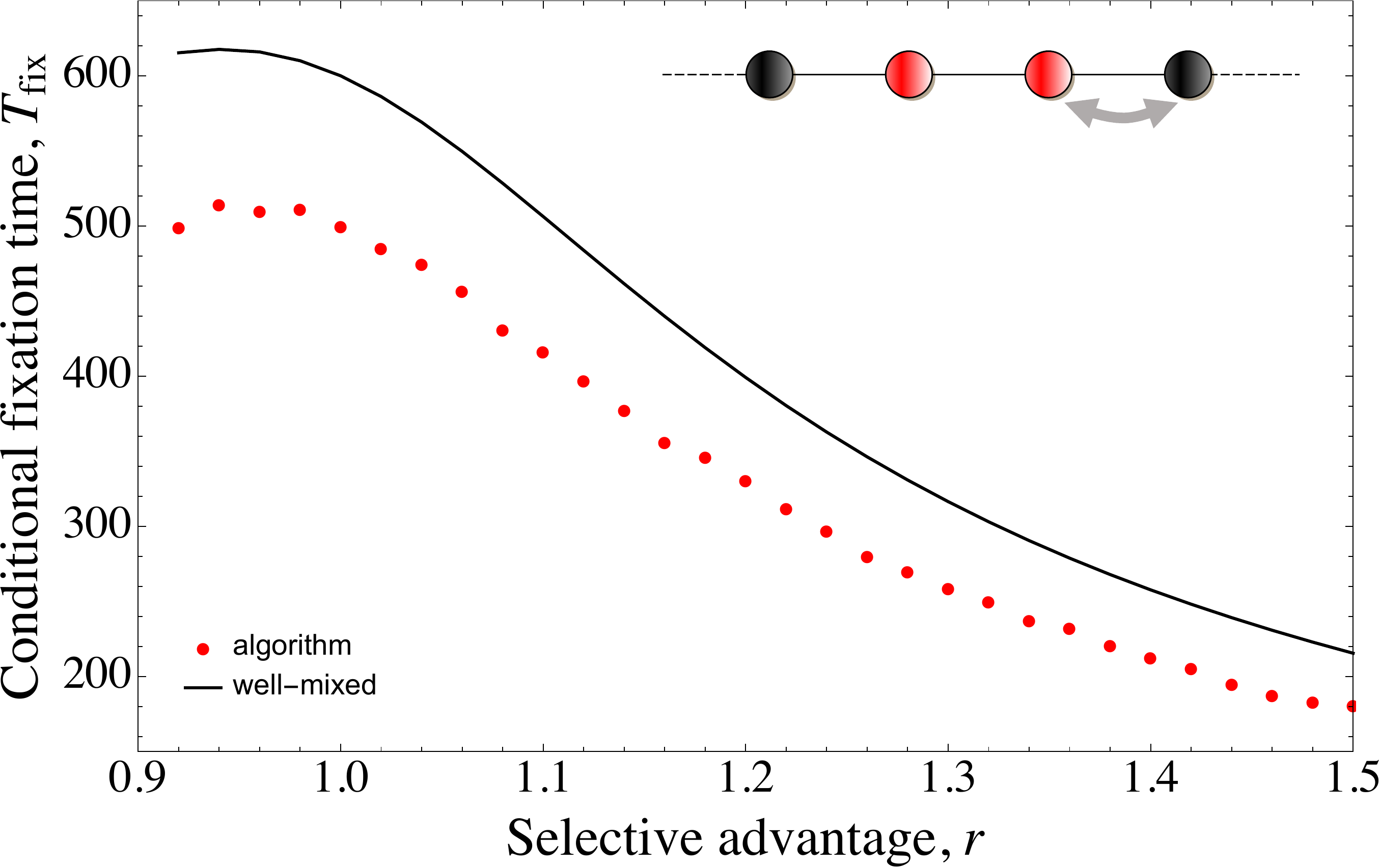}
\caption{\small An algorithm that uses motion to maximize the number of clusters, $c$, of mutants between every update step can accelerate the evolutionary process such that the conditional time to fixation is even less than in the well-mixed case. \label{fig:ALGORITHM}}
\end{figure}
\begin{figure}[ht]
\centering
\includegraphics[width=0.8\columnwidth]{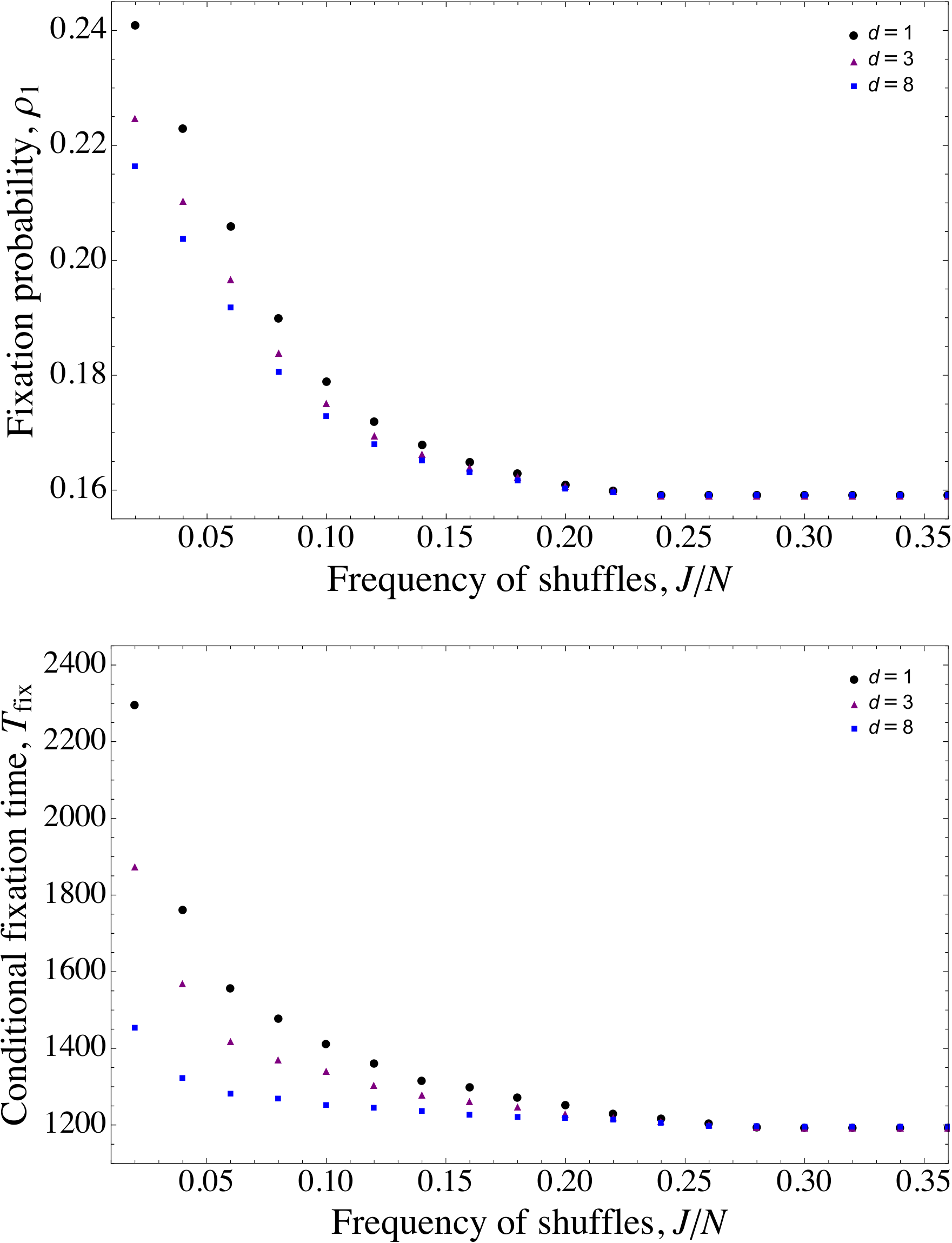}
\caption{\small The effect of multiple shuffles. Fixation probability (top) and conditional fixation time (bottom) are plotted for the death-birth process with a single mutant of relative fitness $r=1.5$ in a cycle with motion. During each update step an individual is selected at random and swaps with another random individual who is at most $d$ loci away. This swapping procedure is repeated $J$ times before a reproduction event occurs. The curves shown have $d=1$ (black circles), $d=3$ (purple triangles), and $d=8$ (blue squares). The $x$-axis denotes $J/N$, where we have arrived at the $y$-value by fixing $J/N$ and averaging over ten values of $N$, to avoid the fact that a certain combination of $d$ and $J$ will have a stronger or weaker impact on evolutionary dynamics depending on the actual size $N$ of the cycle. After a certain amount of shuffling, no further effect is seen; types are randomly distributed with uniform probability across the cycle, corresponding to the percolation limit. \label{fig:DBringJplot}}
\end{figure}

\end{document}